\documentclass[12pt,a4paper]{article}
\usepackage{amsmath}
\usepackage{amsthm}
\usepackage{amssymb}
\usepackage{bbm} 
\usepackage{fancyhdr}
\usepackage[retainorgcmds]{IEEEtrantools}
\usepackage[dvips]{graphicx} 
\usepackage{float}  
\usepackage{caption}
\usepackage[justification=centering]{subcaption}

\newcommand{\E}{\mathbb E}

\newcommand{\beq}{\begin{IEEEeqnarray*}{rCl}}
\newcommand{\eeq}{\end{IEEEeqnarray*}}
\newcommand{\mc}{\IEEEeqnarraymulticol{3}{l}}

\theoremstyle{definition}  \newtheorem{defn}{Definition}
\theoremstyle{plain}	\newtheorem{thm}{Theorem}
\theoremstyle{plain}	\newtheorem{prop}{Proposition}
\theoremstyle{plain}	\newtheorem{cor}{Corollary}
\theoremstyle{plain}  \newtheorem{lem}{Lemma}
\theoremstyle{remark} \newtheorem{rem}{Remark}

\begin{document}

\title{The Merton Problem with a Drawdown Constraint on Consumption}
\author{T. Arun \begin{footnote} {Statistical Laboratory, University of Cambridge, Wilberforce Road, Cambridge \mbox{CB3 0WB}, UK;
arun@statslab.cam.ac.uk.}\end{footnote}}
\date{October 18, 2012}

\maketitle

\begin{abstract}
In this paper, we work in the framework of the Merton problem \cite{M} but we impose a drawdown constraint on the consumption process. This means that consumption can never fall below a fixed proportion of the running maximum of past consumption. In terms of economic motivation, this constraint represents a type of habit formation where the investor is reluctant to let his standard of living fall too far from the maximum standard achieved to date. We use techniques from stochastic optimal control and duality theory to obtain our candidate value function and optimal controls, which are then verified.   
\end{abstract}
\vspace{1cm}

\mbox{drawdown} constraint, duality.

\newpage

\section{Introduction}

The Merton problem -- a question about optimal portfolio selection and consumption in continuous time -- is indeed ubiquitous throughout the mathematical finance literature. Since Merton's seminal paper \cite{M} in 1971, many variants of the original problem have been put forward and extensively studied to address various issues arising from economics. For example, Fleming and Hern\'andez--Hern\'andez \cite{FH} considered the case of optimal investment in the presence of stochastic volatility. Davis and Norman \cite{DN}, Dumas and Luciano \cite{DL}, and more recently Muhle-Karbe and co-authors \cite{CMS}, \cite{GM}, \cite{ML} addressed optimal portfolio selection under transaction costs. Rogers and Stapleton \cite{RS} considered optimal investment under time-lagged trading. Vila and Zariphopoulou \cite{VZ} studied optimal consumption and portfolio choice with borrowing constraints. The effects of different types of habit formation on optimal investment and consumption strategies have been explored in \cite{C}, \cite{I}, and \cite{Mu}.  
\\
\\
A particular class of constrained optimal investment problems that forms an important and recurring theme in mathematical finance is optimal investment under a drawdown constraint. This constraint, roughly speaking, means that a certain parameter has to remain above a fixed proportion of the running maximum of its past values. Drawdown constraints on wealth have been studied by Elie and Touzi \cite{ET}, and Roche \cite{Ro}. Carraro, El Karoui, and Ob{\l}\'oj \cite{CEO}, and Cherny and Ob{\l}\'oj \cite{CO} studied drawdown constraints in more general semimartingale settings via Az\'ema--Yor processes. Grossman and Zhou \cite{GZ} considered the problem of maximising the long-term growth rate of expected utility of final wealth, subject to a drawdown constraint. 
\\
\\
The case we consider in this paper is the Merton problem with a drawdown constraint on consumption. Under this condition, the investor cannot let consumption fall below a fixed proportion of the running maximum of past consumption. In mathematical terms, we have that our consumption at time $t$, $c_t$, satisfies
\[
c_t \geq b \sup_{0 \leq s \leq t} c_s \equiv b \bar{c}_t
\]
for a fixed proportion $0 \leq b \leq 1$. 
\\
\\
In terms of economic motivation, this represents a type of habit formation where once the investor has reached a certain standard of living, he is reluctant to let his standard of living to fall too far from that level. Clearly, the case $b=0$ is just the standard Merton problem, and taking $b=1$ gives the special case where consumption is constrained to be non-decreasing. The $b=1$ case was investigated by Dybvig \cite{D} in 1995, and like the standard Merton problem it is possible to obtain an explicit solution in this case. However, taking $0 < b <1$ gives a continuum of cases between these two extremes where the parameter $b$ in a sense represents the willingness of the investor to sacrifice a proportion of his current standard of living in exchange for greater utility in the long-run. 
\\
\\
To be precise, we consider an agent who can invest in a risk-free bank account and a risky stock modelled by geometric Brownian motion. The agent seeks to maximise the expected infinite horizon discounted utility of consumption by finding the optimal portfolio selection and consumption strategies -- subject to the drawdown constraint on consumption. 
\\
\\
We work with CRRA utility and consider the dual formulation of the problem. The dual problem is significantly easier to handle and has an explicit analytic solution. We invert this to obtain our candidate value function and optimal controls. To prove optimality, we modify the approach of Dybvig \cite{D} (who considered the case where consumption is non-decreasing). The key parameter in this problem is the ratio of the investor's wealth to the running maximum of past consumption. For the optimal solution, we observe four different regions of behaviour based on the value of this parameter. For low values, consumption is restricted to the minimal level possible without violating the drawdown constraint. As the ratio increases, consumption increases with wealth. In the third region, we consume at the highest recorded level of consumption to date while we wait for the ratio to hit a critical level, after which we increase consumption to a new maximum. We specify the boundaries of these regions explicitly, as well as the optimal portfolio selection and consumption rules in each case. 
\\
\\
This paper is organised as follows. In section \ref{marketmodeldraw}, we outline the market model that we will be working in. In sections \ref{Rnot1} and \ref{Ris1}, we provide an informal but intuitive derivation of the value function and optimal controls for $R \neq 1$ and $R=1$, where $R$ represents the investor's coefficient of relative risk aversion. Section \ref{verification} provides a rigorous verification argument to prove the optimality of our conjectured solution. Finally, in section \ref{illposed}, we give an argument to show that, just like in the standard Merton problem, the case we consider here is  ill-posed for $R \leq R^*$ for a certain $0< R^* <1$ which we specify.

\section{Market model} \label{marketmodeldraw}

We work in the framework of the standard Merton problem. Formally, we have a risk-free bank account with constant interest rate, $r$, and a risky stock, $S$, with price dynamics given by
\[
dS_t = S_t(\sigma dW_t + \mu dt)
\]
for constant volatility, $\sigma$, and constant drift, $\mu$, where $(W_t)_{t \geq 0}$ is a standard Brownian motion. Thus our wealth evolves according to the following wealth equation 
\begin{equation} \label{wealth equation}
dw_t = rw_t dt + \theta_t(\sigma dW_t + (\mu - r) dt) - c_t dt
\end{equation}
where
\beq
w_t & = & \text{our wealth at time }t
\\
c_t & = & \text{our consumption at time }t
\\
\theta_t & = & \text{the wealth in the stock at time }t.
\eeq
To make the stock attractive to the investor, we assume that $\mu > r$. We also take $r > 0$ (so we exclude to zero interest rate case) which will in fact turn out to be a necessary condition for the existence of a solution, as shown in Corollary \ref{feasibility}. 
\\
\\
We want to maximize the expected infinite horizon discounted utility of consumption
\begin{equation}
\E \left[\int_0^\infty e^{-\rho t}U(c_t) dt \right]
\end{equation}
subject to a drawdown constraint on consumption
\begin{equation}
c_t \geq b \bar{c}_t \equiv b \sup_{s \leq t} c_s 
\end{equation}
for some $0 < b <1$. Note that taking $b=0$ gives the standard Merton problem \cite{M} and taking $b=1$ gives the non-decreasing consumption case considered by Dybvig \cite{D}. One can check that the analysis put forward in this paper simplifies to the solutions given by Merton and Dybvig for $b \in \{0,1\}$, but to avoid denegerate cases we will restrict our attention to $0 < b < 1$. 
\\
\\
We take the agent's utility function, $U$, to be of constant relative risk aversion (CRRA), that is $U(x) = \frac{x^{1-R}}{1-R}$ for $R \neq 1$, and $U(x) = \log x$ for $R = 1$, where $R$ is a positive real number which represents the investor's coefficient of relative risk aversion. However, for the problem to be well-posed we need to choose $R$ such that $\gamma_M >0$ where 
\begin{equation} \label{gamma_M}
\gamma_M = \frac{1}{R} \left[\rho - (1-R) \left(r + \frac{\kappa^2}{2R} \right) \right]
\end{equation}
and $\kappa = \frac{\mu-r}{\sigma}$. This is equivalent to taking $R>R^*$ for a particular $0 < R^* < 1$ given by
\begin{equation} \label{Rstar}
R^* = \frac{1}{2r} \left[-\left(\rho - r + \frac{\kappa^2}{2} \right) + \sqrt{\left(\rho - r + \frac{\kappa^2}{2} \right)^2 + 2r\kappa^2} \right]
\end{equation}
In section \ref{illposed}, we will show that, as in the standard Merton problem, if we do not have this condition then it is possible to find strategies that give infinite expected utility. 
\\
\\
Lastly, we insist that our investment and consumption strategy, $(\theta, c)$, is \emph{admissible} i.e. $w_t \geq 0$ almost surely for all $t \geq 0$.
\\
\\
In the next section, we will give a systematic but, in some places, informal derivation of the value function and optimal controls for $R \neq 1$. To avoid confusion, we defer the $R=1$ case until section \ref{Ris1}. A rigorous verification argument is given in section \ref{verification}. 

\section{Identifying the optimal controls and the value function for $R \neq 1$} \label{Rnot1}

We call a strategy, $(\theta, c)$, \emph{feasible} if it satisfies the drawdown constraint. We will see that necessary conditions for feasibility are that $r > 0$ and $r w_t \geq b \bar{c}_t$ almost surely for all $t \geq 0$. An intuitive explanation for why this is true is the following. To be able to sustain indefinitely consumption at a rate $c_t \geq b \bar{c}_t$, the consumption would have to be taken from a source of income that is guaranteed, so can only be taken from the interest from the bank account. For this to be possible we need to have $r > 0$ and $r w_t \geq b \bar{c}_t$ for all $t \geq 0$, since the second inequality means that the maximum possible interest that can be gained from wealth is at least the minimum amount that must be consumed. A proof of this statement is given under Corollary \ref{feasibility} in section \ref{verification}. 
\\
\\
We begin as one usually does for problems of this type -- by defining the value function. In contrast to the standard Merton story, our value function, $V(\cdot, \cdot)$, depends on two variables instead of one. Define
\begin{equation} \label{value function defn} 
V(w, \bar{c}) \equiv \sup_{c \geq b \bar{c}, \theta} \E \left[ \int_0^\infty e^{-\rho t}U(c_t) dt \Big|w_0=w, \bar{c}_0= \bar{c} \right].
\end{equation}
Now, let  
\[
Y_t = e^{-\rho t} V(w_t, \bar{c}_t) + \int_0^t e^{-\rho s} U(c_s) ds.
\]
By the Davis--Varaiya Martingale Principle of Optimal Control \cite{DV}, we should have that $Y$ is a supermartingale for all controls, and there exist optimal controls (to be found) such that $Y$ is a true martingale. In what follows, we will use this condition to derive the Hamilton--Jacobi--Bellman (HJB) equation for this problem. We will show that there is only one function that satisfies the HJB equation, subject to appropriate boundary conditions, and we will take this function as our candidate value function and define candidate optimal controls based on this function. In section \ref{verification}, we will verify that our candidate function really is the value function for this problem, and that our candidate optimal controls are in fact optimal.  
\\
\\
By It\^o's formula, 
\beq
e^{\rho t} dY_t & = & \left[ -\rho V + V_w(rw_t + \theta_t(\mu - r) - c_t) + \frac{1}{2} \sigma^2 \theta^2 V_{ww} + U(c_t) \right] dt 
\\
& & +\: V_{\bar{c}}d\bar{c}_t + V_w \theta_t \sigma dW_t
\eeq
where, for example, $V_w$ represents the partial derivative of $V$ with respect to $w$. We deduce that for $Y$ to be a supermartingale for all controls and a martingale under the optimal control, we require
\begin{equation}
V_{\bar{c}} \leq 0
\end{equation}
and when $c=\bar{c}$ we must have $V_{\bar{c}}=0$. We also require
\begin{equation}
\sup_{c \geq b\bar{c}, \theta} \left[ -\rho V + V_w(rw_t + \theta_t(\mu - r) - c_t) + \frac{1}{2} \sigma^2 \theta_t^2 V_{ww} + U(c_t) \right] = 0.
\end{equation}
Thus the Hamilton--Jacobi--Bellman (HJB) equation for this problem is
\begin{equation}
\max\left\{V_{\bar{c}} , \sup_{c\geq b\bar{c}, \theta} \left[ -\rho V + V_w(rw+\theta(\mu-r) -c) + \frac12 \sigma^2 \theta^2V_{ww} + U(c) \right] \right\} = 0.
\end{equation} 
Now, using scaling, we can reduce the number of dimensions of the problem. To do this, let $\mathcal{A}(w, \bar{c})$ be the set of feasible strategies, $(\theta, c)$, starting from initial wealth, $w$, and initial maximum consumption, $\bar{c}$. 
\\
\\
Take $\lambda >0$. From the linearity of wealth dynamics we have that 
\[
(\theta, c) \in \mathcal{A}(\lambda w, \lambda \bar{c})\quad  \Leftrightarrow \quad (\tilde{\theta}, \tilde{c}) \in \mathcal{A}(w, \bar{c})
\] 
where $(\tilde{\theta}, \tilde{c}) = (\theta/\lambda, c/\lambda)$. Now observe that 
\beq
V(\lambda w, \lambda \bar{c}) & = & \sup_{\theta, c} \E \left[ \int_0^\infty e^{-\rho t} \left( \frac{c_t^{1-R}}{1-R} \right) dt \Big| w_0 = \lambda w, \bar{c}_0 = \lambda \bar{c} \right]
\\
& = & \sup_{\tilde{\theta}, \tilde{c}} \E \left[ \int_0^\infty e^{-\rho t} \left( \frac{(\lambda \tilde{c}_t)^{1-R}}{1-R} \right) dt \Big| w_0 = w, \bar{c}_0 = \bar{c} \right]
\\
& = & \lambda^{1-R} \sup_{\tilde{\theta}, \tilde{c}} \E \left[ \int_0^\infty e^{-\rho t} \left( \frac{\tilde{c}_t^{1-R}}{1-R} \right) dt \Big| w_0 = w, \bar{c}_0 = \bar{c} \right]
\\
& = & \lambda^{1-R} V(w, \bar{c}).
\eeq
Thus taking $\lambda = 1/ \bar{c}$ implies that 
\begin{equation}
V(w, \bar{c}) = \bar{c}^{1-R} V \left(\frac{w}{\bar{c}}, 1 \right) \equiv \bar{c}^{1-R} v(x)
\end{equation}
where $x = w/\bar{c}$, and we take the above equation as the definition of $v$, the \emph{scaled value function}. Note that the feasibility condition mentioned at the start of this section now becomes $x \geq b/r$. We have
\begin{eqnarray}
V_w & = & \bar{c}^{\;-R} v'
\\
V_{ww} & = & \bar{c}^{\;-1-R} v''
\\
V_{\bar{c}} & = & \bar{c}^{\;-R} \{(1-R)v - xv' \}.
\end{eqnarray}
Substituting the above into the HJB equation gives
\begin{eqnarray}
\lefteqn{\max \Bigg\{(1-R)v-xv', } \nonumber
\\
& &  \sup_{y, b \leq s \leq 1} \left[ -\rho v + (rx + y(\mu - r) - s)v' + \frac{1}{2} \sigma^2 y^2 v'' + U(s) \right] \Bigg\} = 0
\end{eqnarray}
where $y = \theta/\bar{c}$ and $s = c/\bar{c}$.  
\\
\\
We assume that there is a type of threshold behaviour, which will be verified later. To be precise, we assume that the first term in the HJB equation is equal to zero iff $x \geq a$ and that the second term is equal to zero iff $x < a$ for some $a$ to be determined. The intuitive reasoning for this is that the first term is only zero when we increase $\bar{c}$ which would only happen if $x$ were large. This is because large $x$ means that our wealth is very large compared to the running maximum of past consumption, so we have more than enough wealth to maintain our current maximum level of consumption, so it is in our best interests to raise $\bar{c}$ and increase consumption from then on. 
\\
\\
Consider the region $x \leq a$ first, which corresponds to the second term in the HJB equation. We can divide this into two maximisation problems. The first is
\[
\sup_{y} \left[ y (\mu - r) v' + \frac{1}{2} \sigma^2 y^2 v'' \right]
\]
and differentiating with respect to $y$ gives
\begin{equation} \label{y_optimal}
y = - \frac{\mu - r}{\sigma^2} \frac{v'}{v''}.
\end{equation}
The second maximisation is 
\[
\sup_{b \leq s \leq 1} \{ U(s) - v' s  \}
\]
which has solution
\begin{equation}
s = 
\begin{cases}
1 & \text{for } \quad z_a \leq v' \leq 1
\\
(v')^{-1/R} & \text{for } \quad 1 \leq v' \leq b^{-R}
\\
b & \text{for } \quad b^{-R} \leq v' \leq z_{b/r}
\end{cases}
\end{equation}
where $z_a$ and $z_{b/r}$ are constants to be determined. (The reason for this choice of notation will become clear when we change to dual variables later on.) Here, $z_a$ represents the value of $v'$ at which we decide to increase our maximum consumption, i.e. $z_a = v'(a)$ where $a$ comes from the assumed threshold-type behaviour. Similarly, $z_{b/r} = v'(b/r)$ where the ratio $b/r$ comes from the feasibility condition. 
\\
\\
Putting this all together gives
\[
\begin{array}{rcll}
0 & = & (1-R)v - xv' 				& \text{for } \quad  0 < v' \leq z_a \vspace{3mm} 
\\
0 & = & - \rho v + rxv' - \frac{1}{2} \kappa^2 \frac{(v')^2}{v''} + U(1) - v' & \text{for } \quad z_a \leq v' \leq 1
\vspace{3mm} 
\\
0 & = & - \rho v + rxv' - \frac{1}{2} \kappa^2 \frac{(v')^2}{v''} + \tilde{U}(v') & \text{for } \quad 1 \leq v' \leq b^{-R}
\vspace{3mm} 
\\
0 & = & - \rho v + rxv' - \frac{1}{2} \kappa^2 \frac{(v')^2}{v''} + U(b) - bv' & \text{for } \quad b^{-R} \leq v' \leq z_{b/r}
\end{array}
\]
where $\kappa = (\mu - r)/\sigma$ and $\tilde{U}$ is the dual function of $U$, i.e. 
\[
\tilde{U}(y) = \sup_{b \leq s \leq 1} \{ U(s) - y s  \} = -\frac{y^{1-R'}}{1-R'}
\]
where we define $R' = 1/R$.
\\
\\
It is perhaps clearer to see what is going on if we rewrite the boundaries in terms of $x \equiv w/\bar{c}$. To do this, let $x_z$ be the value of $x$ such that $v'(x) = z$. By definition of the value function \eqref{value function defn}, the function $v'$ is a decreasing function of $x$ so we can rewrite the above system of equations as:
\begin{equation}
\begin{array}{rcll}
0 & = & - \rho v + rxv' - \frac{1}{2} \kappa^2 \frac{(v')^2}{v''} + U(b) - bv' & \text{for } \quad b/r \leq x \leq x_{b^{-R}}
\vspace{3mm} 
\\
0 & = & - \rho v + rxv' - \frac{1}{2} \kappa^2 \frac{(v')^2}{v''} + \tilde{U}(v') & \text{for } \quad x_{b^{-R}} \leq x \leq x_1
\vspace{3mm} 
\\
0 & = & - \rho v + rxv' - \frac{1}{2} \kappa^2 \frac{(v')^2}{v''} + U(1) - v' & \text{for } \quad x_1 \leq x \leq a
\vspace{3mm} 
\\
0 & = & (1-R)v - xv' 				& \text{for } \quad  a \leq x < \infty \vspace{3mm} 
\\
\end{array}
\end{equation}
An intuitive explanation for what is happening in these regions is as follows. 
\\
\\
First consider $b/r \leq x \leq x_{b^{-R}}$, which is the region where $x$ is smallest. At $x = b/r$ we have just enough wealth to maintain the drawdown constraint if we put all our wealth in the bank account, and consume the interest. As $x$ increases until $x = x_{b^{-R}}$ we still consume at the minimum allowed level, $b\bar{c}$, but we have excess wealth which we invest in the risky stock. 
\\
\\
As $x$ increases, we enter the region $x_{b^{-R}} \leq x \leq x_1$. Here our consumption, $c$, increases with $x$ until $c = \bar{c}$ which is the point at which we enter the next region. 
\\
\\
For $x_1 \leq x \leq a$, we are consuming at $c=\bar{c}$ and we keep our consumption constant at this level until $x$ hits a certain critical value, $a$, to be determined.  
\\
\\
In the final region, $a \leq x < \infty$, $x$ is large, and the optimal action here is to immediately increase $\bar{c}$ until $x$ decreases to $a$ which brings us back to the previous region. This ensures that, under this strategy, the set of times spent outside the region $x \in [b/r, a]$ has zero Lebesgue measure. 
\\
\\
As suggested by the above reasoning, we have several boundary conditions at $x = b/r$. At this value of $x$, all our wealth needs to be in the bank account to generate enough interest to maintain the drawdown constraint. If we have non-zero wealth in the stock, the effect of the Brownian motion means that with positive probability, $x$ will fall below $b/r$ which would violate the condition $x \geq b/r$ which is a necessary condition for feasibility. Thus as $x \downarrow b/r$, we must have $v'(x)/v''(x) \rightarrow 0$ which would imply the amount of wealth in the risky stock goes to zero by the form of $y$ given in \eqref{y_optimal}. So our first boundary condition is
\begin{equation}
\frac{v'(x)}{v''(x)} \rightarrow 0 \quad \text{as } x \downarrow b/r.
\end{equation}
Now, if we ever hit $x = b/r$, all our wealth is in the bank account, and the interest generated by our wealth, $rw$, is exactly cancelled by our consumption at level $b \bar{c}$. Hence, our wealth and consumption remain constant, which gives the second boundary condition
\begin{equation} \label{b/r_boundary}
v(b/r) = U(b)/\rho.
\end{equation}
To solve this system of ordinary differential equations subject to the given boundary conditions we transform to dual variables
\begin{eqnarray} 
z & = & v'  \label{dualdefn1}
\\
J(z) & = & \sup_{x > b/r} \{ v(x) - xz \}. \label{dualdefn2}
\end{eqnarray}
Differentiating the above gives
\begin{eqnarray}
J'' & = & -1/v''
\\
J' & = & -x \label{recover x}
\end{eqnarray}
and the system of differential equations becomes 
\begin{equation}
\begin{array}{rcll}
0 & = & (1-R)J + RJ'z  				& \text{for } \quad  0 < z \leq z_a
\\
0 & = & - \rho J + (\rho - r)zJ' + \frac{1}{2} \kappa^2 z^2 J'' + U(1) - z & \text{for } \quad z_a \leq z \leq 1
\\
0 & = & - \rho J + (\rho - r)zJ' + \frac{1}{2} \kappa^2 z^2 J'' + \tilde{U}(z) & \text{for } \quad 1 \leq z \leq b^{-R}
\\
0 & = & - \rho J + (\rho - r)zJ' + \frac{1}{2} \kappa^2 z^2 J'' + U(b) - bz  & \text{for } \quad b^{-R} \leq z \leq z_{b/r}.
\end{array}
\end{equation}
We can also rewrite our two boundary conditions as follows. The first boundary condition becomes
\begin{equation}
zJ''(z) \rightarrow 0 \quad \text{as } z \rightarrow z_{b/r}.
\end{equation}
For the second boundary condition, we need to be more careful. If $z_{b/r} < \infty$, it becomes
\begin{equation}
J(z_{b/r}) = \frac{U(b)}{\rho} - \frac{b}{r} z_{b/r}.
\end{equation}
If, however, $z_{b/r} = \infty$ (which will turn out to be the case), we can rewrite the second boundary condition as
\begin{equation}
\left|J(z) - \frac{U(b)}{\rho} + \frac{b}{r} z \right| \rightarrow 0 \quad \text{as } z \rightarrow \infty. 
\end{equation}
In the first region, $0 < z \leq z_a$, we can solve for $J$ to obtain 
\[
J(z) = A z^{1-R'}
\]
for some constant $A$. 
Next, consider the last three regions. The homogeneous ODE 
\[
0 = - \rho J + (\rho - r)zJ' + \frac{1}{2} \kappa^2 z^2 J''
\] 
has complementary function 
\[
J_c(z) = B_1 z^{-\alpha} + B_2 z^{\beta}
\]
for constants $B_1$, $B_2$ and where $-\alpha < 0 < 1 < \beta$ are the roots of 
\[
Q(t) = \frac{1}{2}\kappa^2 t(t-1) + (\rho - r)t - \rho.
\]
By straightforward verification, we can check that the following are particular solutions for each region 
\[
J_p(z) = 
\begin{cases}
-\frac{1}{r} z + \frac{U(1)}{\rho} & \text{for } \quad z_a \leq z \leq 1 \vspace{3mm}
\\
 \frac{1}{\gamma_M} \tilde{U}(z) & \text{for } \quad 1 \leq z \leq b^{-R} \vspace{3mm}
\\
-\frac{b}{r}z + \frac{U(b)}{\rho} & \text{for } \quad b^{-R} \leq z \leq z_{b/r}.
\end{cases} 
\]
for $\gamma_M$ as defined in \eqref{gamma_M}. 
\\
\\
Thus the general solution for $J$ is 
\begin{equation}
J(z) = 
\begin{cases}
A z^{1-R'} & \text{for } \quad 0 \leq z \leq z_a \vspace{3mm} 
\\
B z^{-\alpha} + C z^{\beta} - \frac{1}{r} z + \frac{U(1)}{\rho} & \text{for } \quad z_a \leq z \leq 1 \vspace{3mm}
\\
D z^{-\alpha} + E z^{\beta}  - \frac{1}{\gamma_M} \frac{z^{1-R'}}{1-R'} & \text{for } \quad 1 \leq z \leq b^{-R} \vspace{3mm}
\\
F z^{-\alpha} + G z^{\beta}  -\frac{b}{r}z + \frac{U(b)}{\rho} & \text{for } \quad b^{-R} \leq z \leq z_{b/r}.
\end{cases} 
\end{equation}
for constants $A$, $B$, $C$, $D$, $E$, $F$, $G$, $z_a$ and $z_{b/r}$ to be determined. 
\\
\\
We guess that $z_{b/r} = \infty$. This makes intuitive sense because $z_{b/r} = v'(b/r)$. If $x$ ever hits $b/r$ then we are stuck at this level, and have no choice but to consume at the minimum allowed level from this point onwards. Thus, it makes sense that any deviation from this point would be significantly more preferable than remaining there, which would give $v'(b/r)= \infty$. Then, the boundary conditions at $z = z_{b/r}$ imply that
\begin{equation}
|F z^{-\alpha} + G z^\beta| \rightarrow 0 \quad \text{as } z \rightarrow \infty
\end{equation}
and 
\begin{equation}
|\alpha (\alpha +1) F z^{-\alpha-1} + \beta (\beta -1 )G z^{\beta-1}| \rightarrow 0 \quad \text{as } z \rightarrow \infty.
\end{equation}

The above boundary conditions, together with equality of the function, and its first and second derivatives, at $z_a$, $1$ and $b^{-R}$ (which is necessary because we are using It\^o's formula) allow us to determine all the constants as given below:

\beq
C & = & \frac{\left( b^{1+R(\beta-1)} - 1 \right)}{\beta(\alpha + \beta)} \left[ \frac{1}{R\gamma_M} (R(\alpha + 1) - 1) - \frac{\alpha + 1}{r} \right]
\eeq

$z_a$ is the solution between 0 and 1 of the equation
\begin{equation} \label{za}
0 = (\alpha + \beta) (R(\beta - 1) + 1) Cz_a^\beta - \frac{(\alpha + 1) z_a}{r} + \frac{\alpha}{\rho} 
\end{equation} 
\beq
A & = & \frac{z_a^{R'-1}}{\gamma_M} \left[ \frac{1}{1-R} - z_a \right]
\vspace{3mm}
\\
B & = & \frac{z_a^\alpha}{(\alpha+\beta)(R(\alpha+1)-1)} \left[ \frac{\beta}{\rho} + \frac{(1-\beta)z_a}{r} \right]
\vspace{3mm}
\\
D & = & B + \frac{1}{\alpha(\alpha+\beta)} \left[ \frac{\beta-1}{r} - \frac{1}{R\gamma_M} (1+R(\beta-1))  \right]
\vspace{3mm}
\\
E & = & \frac{b^{1+R(\beta-1)}}{\beta(\alpha+\beta)} \left[\frac{1}{R\gamma_M} (R(\alpha+1) - 1) - \frac{\alpha+1}{r} \right]
\vspace{3mm}
\\
F & = & B + \frac{\left(1 - b^{1-R(\alpha+1)} \right)}{\alpha(\alpha + \beta)} \left[ \frac{\beta-1}{r} - \frac{1}{R\gamma_M} (1 + R(\beta-1)) \right]
\vspace{3mm}
\\
G & = & 0
\eeq 
Thus, we have a function $J$ which is twice continuously differentiable on $0 < z < \infty$. Note that since we have an explicit form for $J$ we can recover the unknowns $a$, $x_{b^{-R}}$ and $x_1$ using \eqref{recover x} as given below: 
\begin{eqnarray} 
a & = & -J'(z_a) \label{a}
\\
x_{b^{-R}} & = & -J'(b^{-R}) \label{x_b-R}
\\
x_1 & = & -J'(1) \label{x_1}
\end{eqnarray} 
We can take the dual of $J$ to recover $v$ as follows
\[
v(x) = \inf_{0<z<\infty} \{ J(z) + xz \}.
\]
Unfortunately, for $0 <b<1$ it is not possible to obtain $v$ explicitly in all four regions, but we can obtain $v$ explicitly for two of the four regions:
\begin{equation} \label{v explicit Rnot1}
v(x) = 
\begin{cases}
\frac{ \left(x - \frac{b}{r} \right)^{1-R^*}}{1-R^*} (\alpha F)^{R^*} + \frac{ U(b)}{\rho} & \text{for } \quad b/r \leq x \leq x_{b^{-R}} \vspace{3mm} 
\\
U(x) \left[ \frac{1}{-A(1-R') }\right]^{-1/R'} & \text{for } \quad a \leq x < \infty
\end{cases} 
\end{equation}
For the inner two regions, $x_{b^{-R}} \leq x \leq x_1$ and $x_1 \leq x \leq a$ we have to obtain $v$ numerically. 
\\
\\
In the next section, we will show that for $R \neq 1$
\begin{equation} \label{value function}
V(w, \bar{c}) = \bar{c}^{1-R} v(w/\bar{c})
\end{equation}
is the value function for this problem and that the optimal controls are given by
\begin{equation} \label{theta_optimal}
\theta = - \frac{\mu - r}{\sigma^2} \frac{V_w}{V_{ww}} 
\end{equation}
and 
\begin{equation}  \label{c_optimal}
c = 
\begin{cases}
b\bar{c} & \text{for } \quad b/r \leq w/\bar{c} \leq x_{b^{-R}}
\\
(V_w)^{-1/R} & \text{for } \quad x_{b^{-R}} \leq w/\bar{c} \leq x_1
\\
\bar{c} & \text{for } \quad x_1 \leq w/\bar{c} \leq a
\\
w/a & \text{for } \quad a \leq w/\bar{c} < \infty.
\end{cases}
\end{equation} 

We illustrate the optimal strategy and the effect of the drawdown constraint via several figures. 
\\
\\
In Figure \ref{fig1}, we plot the dual function, $J$, and the scaled value function, $v$, as well as the optimal controls, $\theta$ and $c$, all against $x$. 
\\
\\
In Figure \ref{fig2}, we provide a simulation of the stock price followed by plots of $x$ and the optimal controls, all against time, based on this simulation. The horizontal dashed lines in Figure \ref{fig2b} represent the critical values $b/r$, $x_{b^{-R}}$, $x_1$, and $a$ which give the boundaries of the four different regions of behaviour. As $x$ moves between these different regions, we can see the effect on the optimal consumption rule in Figure \ref{fig2d}. In the simulation, consumption initially varies with $x$, then as $x$ increases, consumption is maintained at level $\bar{c}$. As $x$ increases further, $\bar{c}$ is occasionally raised to keep $x \leq a$. Finally as the stock price plummets, $x$ falls as well, so consumption drops until it hits $b \bar{c}$ and is maintained at that level so as not to violate the drawdown constraint. 
\\
\\
Figure \ref{fig3a} shows the scaled value function, $v$, as a function of $x$ for several values of $b$. We clearly see that $v$ decreases as $b$ increases, because increasing $b$ tightens the drawdown constraint, which in turn restricts the class of feasible strategies. Finally, Figure \ref{fig3b} plots $v(x)$ as a function of $b$ for several values of $x$. In this plot, we see once again how increasing $b$ decreases the value of $v(x)$, as one expects. 

\begin{figure}

	\begin{subfigure}{0.5\textwidth}
	\includegraphics[width=\textwidth]{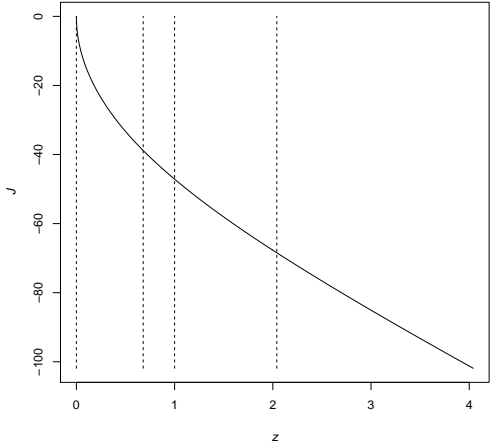}
	\caption{Dual function, $J$, against $z$}
	\end{subfigure}
	~
	\begin{subfigure}{0.5\textwidth}
	\includegraphics[width=\textwidth]{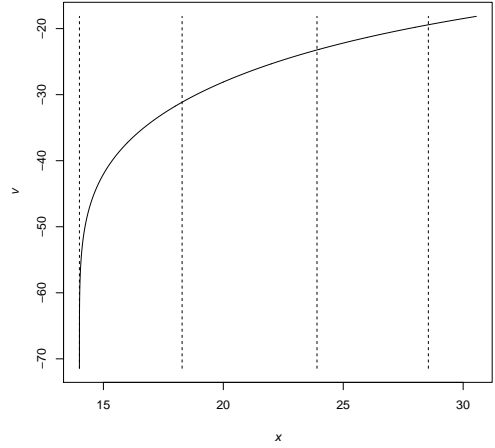}
	\caption{Scaled value function, $v$, against $x$}
	\end{subfigure} 
	
	\begin{subfigure}{0.5\textwidth}
	\includegraphics[width=\textwidth]{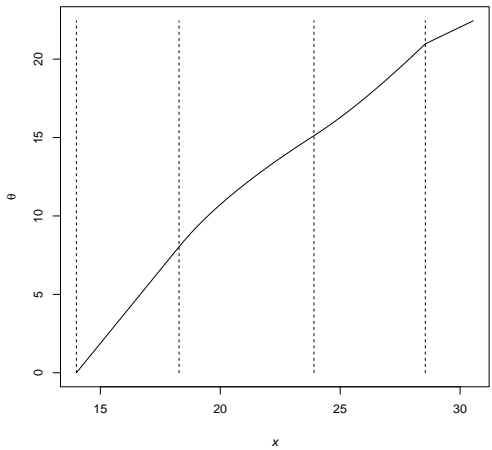}
	\caption{Wealth in stock, $\theta$, under the optimal control, against $x$}
	\end{subfigure}
	~
	\begin{subfigure}{0.5\textwidth}
	\includegraphics[width=\textwidth]{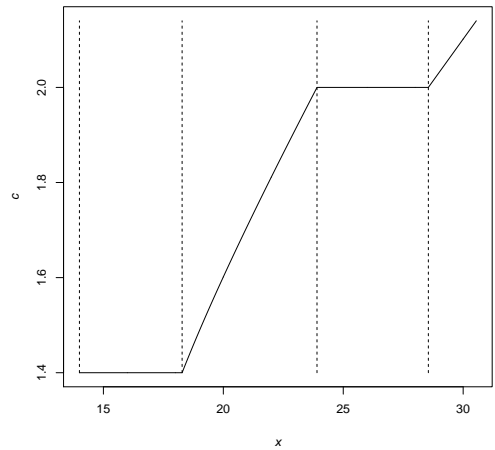}
	\caption{Optimal consumption strategy, $c$, against $x$}
	\end{subfigure} 
	
\caption{The vertical dashed lines represent the critical values $b/r$, $x_{b^{-R}}$, $x_1$, and $a$ which give the boundaries of the four different regions of behaviour. For both graphs we take $b=0.7$, $\bar{c}_0=2$, $R=2$, $\rho=0.02$, $r=0.05$, $\sigma=0.35$, and $\mu=0.14$. } \label{fig1}
\end{figure}

\begin{figure} 
	
	\begin{subfigure}{0.5\textwidth}
	\includegraphics[width=\textwidth]{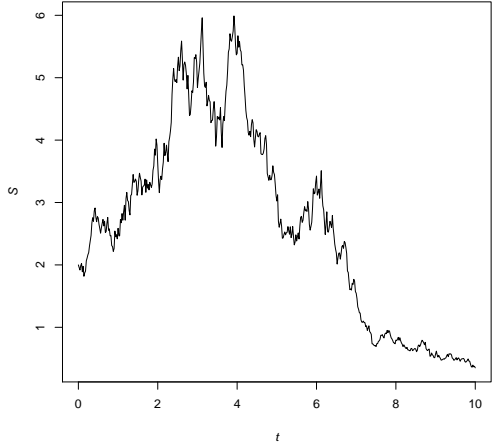}
	\caption{Stock price, $S$, against $t$}
	\end{subfigure}
	~
	\begin{subfigure}{0.5\textwidth}
	\includegraphics[width=\textwidth]{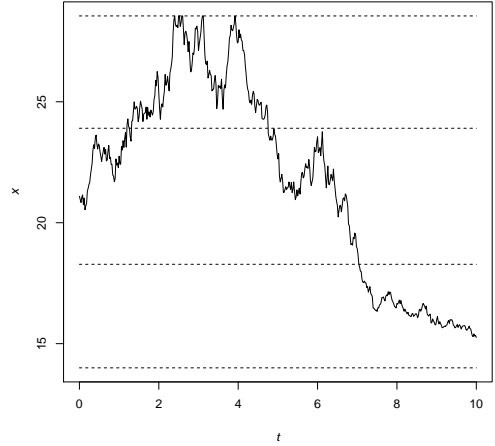}
	\caption{$x$ against $t$} \label{fig2b}
	\end{subfigure}
	
	\begin{subfigure}{0.5\textwidth}
	\includegraphics[width=\textwidth]{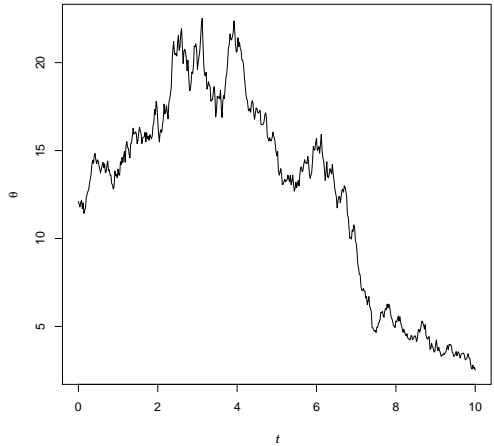}
	\caption{Wealth in stock, $\theta$, against $t$}
	\end{subfigure} 
	~
	\begin{subfigure}{0.5\textwidth}
	\includegraphics[width=\textwidth]{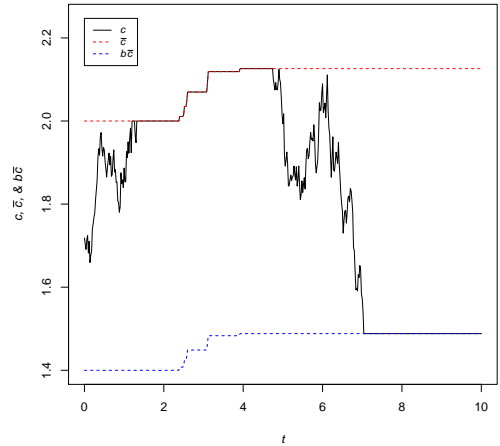}
	\caption{$c$, $\bar{c}$, \& $b\bar{c}$ against $t$} \label{fig2d}
	\end{subfigure} 
	
\caption{The above shows a simulation of the stock price and plots of $x$ and the optimal controls, $\theta$ and $c$, against time, $t$, based on this simulation. In Figure \ref{fig2b}, the horizontal dashed lines represent the critical values $b/r$, $x_{b^{-R}}$, $x_1$, and $a$ which give the boundaries of the four different regions of behaviour. For all three graphs we take $b=0.7$, $\bar{c}_0=2$, $R=2$, $\rho=0.02$, $r=0.05$, $\sigma=0.35$, and $\mu=0.14$. } \label{fig2}
\end{figure}

\begin{figure} 
	\begin{center}
	\begin{subfigure}{0.6\textwidth} 
	\includegraphics[width=\textwidth]{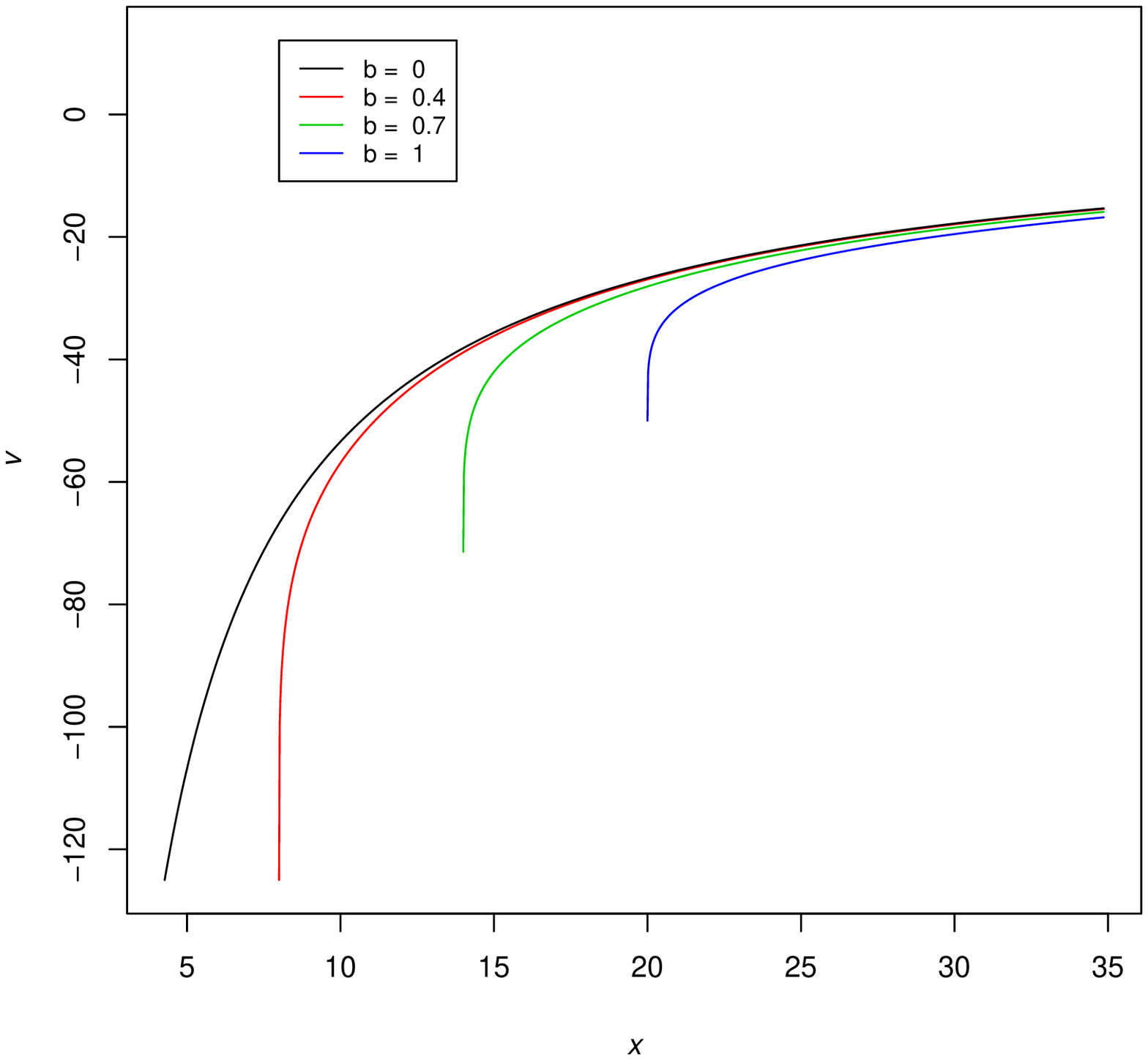}
	\caption{Scaled value function, $v$, against $x$ for several values of $b$} \label{fig3a}
	\end{subfigure}
	
	\begin{subfigure}{0.6\textwidth} 
	\includegraphics[width=\textwidth]{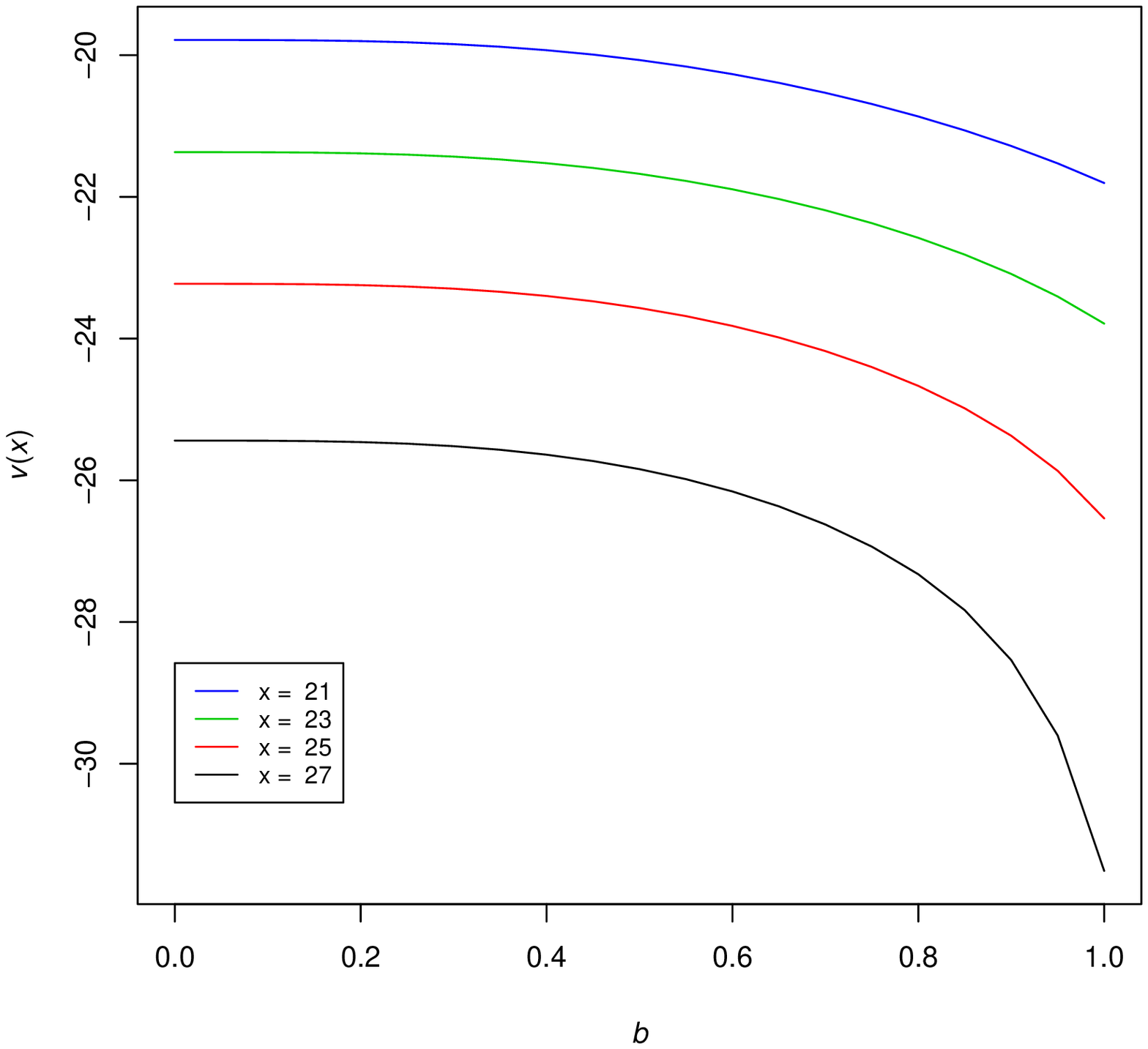}
	\caption{$v(x)$ against $b$ for several values of $x$} \label{fig3b}
	\end{subfigure} 
	\end{center}
			
\caption{In the above we take $\bar{c}_0=2$, $R=2$, $\rho=0.02$, $r=0.05$, $\sigma=0.35$, and $\mu=0.14$. }
\end{figure}

\section{Identifying the optimal controls and the value function for $R = 1$} \label{Ris1}

Now, we consider CRRA utility for $R=1$, that is we take our utility function to be $U(x) = \log x$, and solve the HJB equation in this case. The main difference is that we have a different scaling result. As before, let $\mathcal{A}(w, \bar{c})$ be the set of feasible strategies, $(\theta, c)$, starting from initial wealth, $w$, and initial maximum consumption, $\bar{c}$. 
\\
\\
Take $\lambda >0$. From the linearity of wealth dynamics we have that 
\[
(\theta, c) \in \mathcal{A}(\lambda w, \lambda \bar{c}) \quad \Leftrightarrow \quad (\tilde{\theta}, \tilde{c}) \in \mathcal{A}(w, \bar{c})
\] 
where $(\tilde{\theta}, \tilde{c}) = (\theta/\lambda, c/\lambda)$. Now observe that 
\beq
V(\lambda w, \lambda \bar{c}) & = & \sup_{\theta, c} \E \left[ \int_0^\infty e^{-\rho t} \log c_t dt \Big| w_0 = \lambda w, \bar{c}_0 = \lambda \bar{c} \right]
\\
& = & \sup_{\tilde{\theta}, \tilde{c}} \E \left[ \int_0^\infty e^{-\rho t} \log (\lambda \tilde{c}_t) dt \Big| w_0 = w, \bar{c}_0 = \bar{c} \right]
\\
& = & \frac{\log \lambda}{\rho} +  \sup_{\tilde{\theta}, \tilde{c}} \E \left[ \int_0^\infty e^{-\rho t}  \log \tilde{c}_t dt \Big| w_0 = w, \bar{c}_0 = \bar{c} \right]
\\
& = & \frac{\log \lambda}{\rho} +  V(w, \bar{c}).
\eeq
Thus taking $\lambda = 1/ \bar{c}$ implies that 
\begin{equation}
V(w, \bar{c}) = V\left(\frac{w}{\bar{c}}, 1 \right) + \frac{\log \bar{c}}{\rho} \equiv v(x) + \frac{\log \bar{c}}{\rho}
\end{equation}
where $x = w/\bar{c}$, and we take the above equation as the definition of $v$. We have
\begin{eqnarray}
V_w & = & \bar{c}^{\;-1} v'
\\
V_{ww} & = & \bar{c}^{\;-2} v''
\\
V_{\bar{c}} & = & \bar{c}^{\;-1} \left(\frac{1}{\rho} - xv' \right).
\end{eqnarray}
Now switch to dual variables, as we did before. Let
\begin{eqnarray}
z & = & v'
\\
J(z) & = & \sup_{x > b/r} \{ v(x) - xz \}.
\end{eqnarray}
Apart from this, the analysis is exactly the same as in the previous section, so we omit the details and present the final result. Our solution of the HJB equation for log utility is of the form
\begin{equation}
J(z) = 
\begin{cases}
-\frac{1}{\rho} \log z + A & \text{for } \quad 0 \leq z \leq z_a \vspace{3mm} 
\\
B z^{-\alpha} + C z^{\beta} - \frac{1}{r} z  & \text{for } \quad z_a \leq z \leq 1 \vspace{3mm}
\\
D z^{-\alpha} + E z^{\beta}  - \frac{1}{\rho} (\log z +1) - \frac{1}{\rho^2} (\rho - r - \frac{1}{2} \kappa^2) & \text{for } \quad 1 \leq z \leq 1/b \vspace{3mm}
\\
F z^{-\alpha} + G z^{\beta}  -\frac{b}{r}z + \frac{ \log b}{\rho} & \text{for } \quad 1/b \leq z \leq \infty.
\end{cases} 
\end{equation}
for constants $A$, $B$, $C$, $D$, $E$, $F$, $G$, and $z_a$ as given below.
\[
C = \frac{(b^\beta - 1)}{\beta(\alpha + \beta)} \left[ \frac{\alpha}{\rho} - \frac{\alpha+1}{r} \right]
\]
$z_a$ is the solution between 0 and 1 of the equation
\[
0 = \beta(\alpha+\beta)C z_a^\beta - \frac{(\alpha+1)z_a}{r} + \frac{\alpha}{\rho}
\]
\beq
A & = & \frac{1}{\alpha(\alpha+ \beta)} \left[\frac{\beta}{\rho} - (\beta-1) \frac{z_a}{r} \right] + \frac{1}{\beta(\alpha+\beta)} \left[ (\alpha+1)\frac{z_a}{r} - \frac{\alpha}{\rho} \right] - \frac{z_a}{r} + \frac{\log z_a}{\rho} \vspace{3mm}
\\
B & = & \frac{z_a^\alpha}{\alpha(\alpha+\beta)} \left[ \frac{\beta}{\rho} - ( \beta-1) \frac{z_a}{r} \right] \vspace{3mm}
\\
D & = & B + \frac{1}{\alpha(\alpha+ \beta)} \left[ \frac{\beta-1}{r} - \frac{\beta}{\rho} \right] \vspace{3mm}
\\
E & = & \frac{b^\beta}{\beta(\alpha+\beta)} \left[ \frac{\alpha}{\rho} - \frac{\alpha+1}{r} \right] \vspace{3mm}
\\
F & = & B + \frac{\left(b^{-\alpha}-1\right)}{\alpha(\alpha+\beta)} \left[ \frac{\beta}{\rho} - \frac{\beta-1}{r} \right] \vspace{3mm}
\\
G & = & 0. 
\eeq
And as in section \ref{Rnot1}, we can recover $a$, $x_{b^{-R}}$, and $x_1$ as given below: 
\begin{eqnarray} 
a & = & -J'(z_a) \label{1a}
\\
x_{b^{-R}} & = & -J'(1/b) \label{1x_b-R}
\\
x_1 & = & -J'(1) \label{1x_1}
\end{eqnarray} 
We can take the dual of $J$ to recover $v$ as follows
\[
v(x) = \inf_{0<z<\infty} \{ J(z) + xz \}.
\]
As in the $R \neq 1$ case, it is only possible to invert $J$ explicitly in two of the four regions, as given below. 
\begin{equation} \label{v explicit R=1}
v(x) = 
\begin{cases}
\frac{ \left(x - \frac{b}{r} \right)^{1-R^*}}{1-R^*} (\alpha F)^{R^*} + \frac{ \log b}{\rho} & \text{for } \quad b/r \leq x \leq x_{1/b} \vspace{3mm} 
\\
\frac{1}{\rho} \left( \log x + 1 + \log \rho \right) + A & \text{for } \quad a \leq x < \infty
\end{cases} 
\end{equation}
For the inner two regions, $x_{b^{-R}} \leq x \leq x_1$ and $x_1 \leq x \leq a$ we have to obtain $v$ numerically.
\\
\\
In the next section, we will show that for $R=1$
\begin{equation} \label{1 value function}
V(w, \bar{c}) = v(x) + \frac{\log \bar{c}}{\rho}
\end{equation}
is the value function for this problem and that the optimal controls are given by
\begin{equation} \label{1 theta_optimal}
\theta = - \frac{\mu - r}{\sigma^2} \frac{V_w}{V_{ww}} 
\end{equation}
and 
\begin{equation} \label{1 c_optimal}
c = 
\begin{cases} 
b\bar{c} & \text{for } \quad b/r \leq w/\bar{c} \leq x_{1/b}
\\
1/V_w & \text{for } \quad x_{1/b} \leq w/\bar{c} \leq x_1
\\
\bar{c} & \text{for } \quad x_1 \leq w/\bar{c} \leq a
\\
w/a & \text{for } \quad a \leq w/\bar{c} < \infty.
\end{cases}
\end{equation}

\section{Verification argument} \label{verification}
We modify the argument of Dybvig \cite{D} to prove optimality for our conjectured solution. First, we obtain necessary conditions for feasibility -- that is, we must have $rw_t \geq b\bar{c}_t$ almost surely and $r > 0$. In what follows, let $\E_\tau[\cdot] = \E[\cdot|\mathcal{F}_\tau]$, where $(\mathcal{F}_t)_{t \geq 0}$ represents the filtration generated by the stock price, $S$, or equivalently, by the Brownian motion, $W$. 
\\
\\
We will need the following lemma. 

\begin{lem}\ \\
For all feasible strategies, and for all $\tau \geq 0$,
\[
\E_\tau \left[ \int_{t=\tau}^\infty \frac{\zeta_t c_t}{\zeta_\tau} dt \right] \leq w_\tau
\]
almost surely, where $\zeta_t \equiv \exp (-rt - \frac12 \kappa^2 t - \kappa W_t)$ is the state-price density and where $\kappa = \frac{\mu-r}{\sigma}$, as defined previously.
\end{lem}
\begin{proof}
$(\zeta_t)_{t \geq 0}$ is a strictly positive process and by It\^o's formula, 
\[
d\zeta_t = \zeta_t (-rdt - \kappa dW_t).
\]
Define for $t \geq \tau$
\[
Z_t = \frac{\zeta_t w_t}{\zeta_\tau} + \int_\tau^t \frac{\zeta_s c_s}{\zeta_\tau} ds.
\]
By It\^o's formula,
\beq
dZ_t & = & \frac{\zeta_tc_t}{\zeta_\tau}dt + \frac{\zeta_t}{\zeta_\tau} \left( w_trdt + \theta_t((\mu-r)dt + \sigma dW_t) - c_tdt \right) \\
& & +\: \frac{\zeta_t w_t}{\zeta_\tau} \left(-rdt - \kappa dW_t \right) + \frac{\zeta_t}{\zeta_\tau} \left(- \kappa \theta_t \sigma \right) dt 
\\
& = & \frac{\zeta_t}{\zeta_\tau} \left(\sigma \theta_t - \kappa w_t \right) dW_t.
\eeq
Hence, $Z$ is a positive local martingale which implies that $Z$ is a supermartingale. Finally, using Fatou's Lemma gives 
\beq
\E_\tau \left[ \int_{s=\tau}^\infty \frac{\zeta_s c_s}{\zeta_\tau}ds \right] & \leq & \lim_{t \rightarrow \infty} \E_\tau\left[ \int_{s = \tau}^t \frac{\zeta_s c_s}{\zeta_\tau} ds \right]
\\
& \leq & \lim_{t \rightarrow \infty} \E_\tau \left[ \int_{s=\tau}^t \frac{\zeta_s c_s}{\zeta_\tau} ds + \frac{\zeta_t w_t}{\zeta_\tau} \right]
\\
& = & \lim_{t \rightarrow \infty} \E_\tau [Z_t] 
\\
& \leq & Z_\tau \quad \text{since $Z$ is a supermartingale}
\\
& = & w_\tau.
\eeq
\end{proof}
Using this lemma, we obtain the following corollary which gives necessary conditions for feasibility. 
\begin{cor} \label{feasibility} \ \\
For the Merton problem with a drawdown constraint on consumption to have a solution, we require $r > 0$ and we must have $w_\tau \geq \frac{b \bar{c}_\tau}{r}$ almost surely for all $\tau \geq 0$. 
\end{cor}
\begin{proof}
Fix $\tau \geq 0$. From the previous lemma and the drawdown constraint we have 
\beq
w_\tau & \geq & \E_\tau\left[ \int_{t=\tau}^\infty \frac{\zeta_t c_t}{\zeta_\tau} dt \right]
\\
& \geq & \E_\tau \left[ \int_{t = \tau}^\infty \frac{\zeta_t b \bar{c}_\tau}{\zeta_\tau} dt \right]
\\
& = & \int_{t = \tau}^\infty \E_\tau \left[ \exp \left[ \left(-r-\frac12 \kappa^2 \right) (t -\tau) - \kappa (W_t - W_\tau) \right] b\bar{c}_\tau dt \right]
\\
& = & \int_{t=\tau}^\infty \exp(-r(t-\tau)) b \bar{c}_\tau dt
\\
& = & 
\begin{cases}
\frac{b \bar{c}_\tau}{r} & \text{ for } r > 0
\\
+\infty & \text{otherwise}
\end{cases}
\eeq
where the exchange of the order of integration and expectation is valid because the integrand is non-negative.
\end{proof}

This makes precise the intuitive argument given earlier about why we can restrict our attention to the feasible region $\{(w,\bar{c}): w/\bar{c} \geq b/r \}$. For our conjectured optimal controls (\eqref{theta_optimal}, \eqref{c_optimal} or \eqref{1 theta_optimal}, \eqref{1 c_optimal}), we can further restrict our attention to the region $\{(w,\bar{c}): b/r \leq w/\bar{c} \leq a \}$, because our consumption rule (\eqref{c_optimal} or \eqref{1 c_optimal}) implies that the set of times spent outside this region has zero Lebesgue measure.  
\\
\\
We are now ready to state our verification theorem. Note that we take $\bar{c}_0 > 0$. As mentioned in Dybvig \cite{D}, we could take $\bar{c}_0=0$ without too much difficulty. However, this would only yield a slight increase in generality but would require dealing with many extra cases. 

\begin{thm}\label{verify}\ \\
Given fixed initial conditions $(w_0, \bar{c}_0)$ with $w_0/\bar{c}_0 \geq b/r$ and $\bar{c}_0 > 0$, the Merton problem with a drawdown constraint on consumption has value function $V(w,\bar{c})$ as defined in \eqref{value function} or \eqref{1 value function}. The optimal controls are 
\[
\theta_t = -\frac{\mu-r}{\sigma^2} \frac{V_w}{V_{ww}} 
\]
and 
\[
c_t = 
\begin{cases}
b\bar{c}_t & \text{for } \quad b/r \leq w_t/\bar{c}_t \leq x_{b^{-R}}
\\
(V_w)^{-1/R} & \text{for } \quad x_{b^{-R}} \leq w_t/\bar{c}_t \leq x_1
\\
\bar{c}_t & \text{for } \quad x_1 \leq w_t/\bar{c}_t \leq a
\\
w_t/a & \text{for } \quad a \leq w_t/\bar{c}_t < \infty.
\end{cases}
\]
for constants $a$, $x_{b^{-R}}$ and $x_1$ as defined in \eqref{a}, \eqref{x_b-R}, \eqref{x_1}, or \eqref{1a}, \eqref{1x_b-R}, \eqref{1x_1}.
\end{thm}
We will prove this via a series of lemmas. We will need the following definition.

\begin{defn}
We say that a process $X$ is a \emph{local supermartingale} if there exists a sequence of stopping times $\tau_n$ with $\tau_n \uparrow \infty$ almost surely, such that for each $n \geq 0$ we have that $(X_{t \wedge \tau_n})_{t \geq 0}$ is a supermartingale. 
\end{defn}

\begin{rem}
Clearly, if $X_t = M_t + A_t$ for $M$ a local martingale and $A$ a non-increasing process, then $X$ is a local supermartingale. 
\end{rem} 

\begin{lem}\label{local}\ \\
Let 
\[
Y_t = \int_0^t e^{-\rho s} U(c_s) ds + e^{-\rho t} V(w_t, \bar{c}_t).
\]
Then for any feasible strategy, $(\theta, c)$, $Y$ is a local supermartingale and for the proposed optimal control, $Y$ is a local martingale. 
\end{lem}
\begin{proof}
This is essentially true by construction because we chose $V$ to be the solution of the HJB equation. However, there are a few things left to verify. We need to check that $V_{\bar{c}} \leq 0$, $V_w \geq 0$ and $V_{ww} < 0$ to ensure that the drift term in the It\^o expansion of $Y$ is non-positive for all feasible strategies and is identically zero for the conjectured optimal control. By the definition of $J$ in \eqref{dualdefn2}, it is sufficient to show that $(1-R)J+RJ'z \leq 0$, $J' \leq 0$ and $J''>0$. This is a straightforward but surprisingly tedious exercise and we omit the details. 
\end{proof}

The next step is to strengthen the conclusion of the above lemma from local (super)martingale to (super)martingale. To do this, we first need to prove a result about the wealth process, $(w_t)_{t \geq 0}$. 

\begin{lem} \label{wealth}\ \\
Fix $\bar{c}_0 >0$ and $p \neq 0$. Given any feasible strategy, $(\theta, c)$, we have
\begin{eqnarray} 
w_t^p & = & w_0^p \exp \left(\int_{s=0}^t p\left(r + \frac{\theta_s}{w_s} (\mu -r) - \frac{c_s}{w_s} + \frac12 (p-1) \left(\frac{\theta_s}{w_s} \right)^2 \sigma^2 \right) ds \right) \nonumber
\\
& & \times\: \exp \left( \int_{s=0}^t p\frac{\theta_s}{w_s} \sigma dW_s - \frac12 \int_{s=0}^t p^2 \left(\frac{\theta_s}{w_s} \right)^2 \sigma^2 ds \right) \label{2}
\end{eqnarray}
where the second exponential term is a stochastic exponential (or Dol\'eans-Dade exponential) which is a non-negative local martingale thus a supermartingale. For the proposed optimal control, this stochastic exponential is, in fact, a true martingale. Furthermore, for the optimal control, there exists a constant $\tilde{b}$ depending on $p$ and the parameters of the problem such that 
\begin{equation} \label{3}
\E\left[ w_t^p \right] \leq w_0^p \exp(\tilde{b}t). 
\end{equation}
\end{lem}
\begin{proof}
It\^o's formula tells us that $d \log (w^p) = \frac{p}{w} dw - \frac12 \frac{p}{w^2} d\langle w \rangle$. Note that It\^o's formula is valid because the logarithm and power functions are smooth over the relevant domain since feasibility implies that $w_t \geq b \bar{c}_0/r >0$ because $\bar{c}_0 >0$ by assumption. Integrating the expression for $d\log (w^p)$ and substituting in the wealth equation \eqref{wealth equation} gives the form for $w_t^p$ in \eqref{2}. 
\\
\\
An application of It\^o's formula gives that the second term in \eqref{2} is a local martingale. Since it is clearly non-negative, it is a supermartingale. To show that this is a martingale for the conjectured optimal control, it is sufficient to show that $\theta_s/w_s$ takes values in a compact set of the form $[0, M]$ for some constant $M > 0$. This would imply that Novikov's criterion (page 198, \cite{KS}) is satisfied which would imply that it is a martingale. We have that $\frac{\theta_s}{w_s} = \frac{\theta_s}{\bar{c}_s} \times \frac{\bar{c}_s}{w_s}$. As mentioned before, for the conjectured optimal control, $x_s \equiv w_s/\bar{c}_s$ takes values in the compact set $[b/r, a]$ which implies that  $\bar{c}_s/w_s \in [1/a, r/b]$. Now, to deal with the $\theta_s/\bar{c}_s$ term, first we will show that for the conjectured optimal control, it is a continuous function of $x$ for $x \in [b/r, a]$. For $R \neq 1$ and $R=1$, we have by \eqref{y_optimal}, \eqref{dualdefn1} and \eqref{dualdefn2} that
\[
\frac{\theta_s}{\bar{c}_s} = \frac{\mu - r}{\sigma^2} zJ''.
\]
This is continuous for $x \in (b/r, a]$, or equivalently $z \in [z_a, \infty)$, since by construction, $J''$ is continuous in this region. So the only thing to check is continuity at $x = b/r$. At this critical value of $x$, we set $\theta = 0$ and place all our wealth in the bank account, so we need to check that $\theta \rightarrow 0$ as $x \downarrow b/r$. By the above equation, this is equivalent to checking that $zJ'' \rightarrow 0$ as $z \uparrow \infty$. But for $z \geq b^{-R}$, we have 
\[
zJ'' = \alpha(\alpha+1) F z^{-\alpha - 1} \rightarrow 0 \text{ as } z \rightarrow \infty
\]
as required. Thus, $\theta_s/\bar{c}_s$ is a continuous function of $x$ for $x \in [b/r, a]$ and since $x$ takes values in a compact set, we have that $\theta_s/\bar{c}_s$ takes values in a compact set of the form $[0, \tilde{M}]$ for some constant $\tilde{M} > 0$. Thus, we have that $\theta_s/w_s \in [0,M]$ for some constant $M>0$ as desired. Hence, by Novikov's criterion, we have that the stochastic exponential is a true martingale. 
\\
\\
Finally we need to prove the stated bound on $\E\left[ \frac{w_t^p}{p} \right]$. We just showed that $\theta_s/w_s$ takes values in a compact set of the form $[0, M]$ but a similar result is true for $c_s/w_s$. Indeed, observe that $\frac{c_s}{w_s} = \frac{c_s}{\bar{c}_s} \times \frac{\bar{c}_s}{w_s}$. The first term is clearly bounded between 0 and 1, and we showed above that the second term takes values in a compact set. Because of this, the integrand in the first exponential in \eqref{2} is bounded and if we let $\tilde{b}$ be an upper bound for it, \eqref{3} follows by the martingale property of the second term (the stochastic exponential) in \eqref{2}. 
\end{proof}

With the above result in hand, we can strengthen the conclusion of Lemma \ref{local} to: 

\begin{lem}\label{nonlocal} \ \\
Fix $\bar{c}_0 > 0$. Let 
\[
Y_t = \int_0^t e^{-\rho s} U(c_s) ds + e^{-\rho t} V(w_t, \bar{c}_t).
\]
Then for any feasible strategy, $(\theta, c)$, $Y$ is a supermartingale and for the proposed optimal control, $Y$ is a martingale. 
\end{lem}
\begin{proof}
For any feasible strategy, Lemma \ref{local} implies that $Y$ is a local supermartingale. It is enough to show that $Y$ is bounded below, because it is easy to see that any local supermartingale bounded below is a supermartingale. Note that the fact that $V_w \geq 0$ (see proof of Lemma \ref{local}) together with the boundary condition \eqref{b/r_boundary} implies that 
\[
V(w_t, \bar{c}_t) \geq V\left(\frac{b\bar{c}_t}{r}, \bar{c}_t \right) = U\left(b\bar{c}_t \right)/\rho \geq U(b \bar{c}_0)/\rho > -\infty.
\] 
Hence, $V$ is bounded below. To show that $Y$ is bounded below, observe that 
\beq
Y_t & = &  \int_0^t e^{-\rho s} U(c_s) ds + e^{-\rho t} V(w_t, \bar{c}_t)
\\
& \geq & \int_0^t e^{-\rho s} U(b\bar{c}_0) ds + e^{-\rho t} U(b\bar{c}_0)/\rho
\\
& = & \frac{U(b \bar{c}_0)}{\rho}
\eeq
which gives that $Y$ is a supermartingale. 
\\
\\
Now consider the proposed optimal control. We know from Lemma \ref{local} that $Y$ is a local martingale under this control. To show that $Y$ is a martingale, it is enough to show that 
\[
\E\langle Y \rangle_t < \infty
\]
for all $t \geq 0$ as this implies the local martingale $Y$ is in fact a true martingale (see Corollary 1.25 in \cite{RY}). We have that under the conjectured optimal control
\[
dY_t = e^{-\rho t} V_w \theta_t \sigma dW_t
\]
where $\theta_t = -\frac{\mu-r}{\sigma^2} \frac{V_w}{V_{ww}}$ (for both $R \neq 1$ and $R = 1$) hence we obtain
\[
dY_t = -\kappa e^{-\rho t} \frac{V_w^2}{V_{ww}} dW_t 
\]
where $\kappa = \frac{\mu-r}{\sigma}$ as defined previously. First recall that (for both $R \neq 1$ and $R = 1$)
\beq
V_w & = & \bar{c}^{\; -R} v'
\\
V_{ww} & = & \bar{c}^{\; -1-R} v''
\\
J'' & = & 1/v''
\\
J' & = & -x
\\
z & = & v'
\eeq
hence 
\[
\frac{V_w^2}{V_{ww}} = \frac{ \bar{c}^{\; -2R} (v')^2}{\bar{c}^{\; -1-R} v''} = - \bar{c}^{\; 1-R} z^2 J''. 
\]
Now, under the conjectured optimal control we have  $z_a \leq z < \infty$, and $z^2J''$ is continuous in this region by construction. For $z_a \leq z \leq b^{-R}$, $z^2J''$ is bounded, since a continuous function on a compact set is bounded. For the final region, $b^{-R} \leq z < \infty$, we have
\[
z^2 J'' = \alpha(\alpha+1)Fz^{-\alpha} \rightarrow 0 \text{ as } z \rightarrow \infty 
\]
because $-\alpha < 0$. Hence, $z^2 J''$ is bounded on $b^{-R} \leq z < \infty$ as well. So $z^2J''$ is bounded on the whole interval $z_a \leq z < \infty$, say
\[
|z^2 J''| \leq K
\]
for some constant $K > 0$. We have
\[
d\langle Y \rangle_t = e^{-2 \rho t} \kappa^2 \bar{c}_t^{\; 2(1-R)} (z^2J'')^2 dt
\]
which gives 
\begin{eqnarray}
\E \langle Y \rangle_t & = & \E \int_0^t e^{-2\rho s} \kappa^2 \bar{c}_s^{\; 2(1-R)} (z^2 J'')^2 ds \nonumber
\\
& \leq & K^2 \kappa^2 \int_0^t \E\left(\bar{c}_s^{\; 2(1-R)}\right) ds \label{finite qvariation}
\end{eqnarray}
where the use of Fubini's Theorem is justified because the integrand is positive. 
\\
\\
Recall that we require $R > R^*$, as explained in section \ref{marketmodeldraw}, which gives us the three following cases. 
\begin{itemize}
\item $R^* < R < 1$: We have 
\[
\bar{c}_s \leq \frac{rw_s}{b}
\]
from the feasibility condition in Corollary \ref{feasibility}. This implies that 
\[
\bar{c}_s^{\; 2(1-R)} \leq \left(\frac{r}{b} \right)^{2(1-R)} w_s^{2(1-R)}
\]
which gives
\beq
\E\left(\bar{c}_s^{\; 2(1-R)} \right) & \leq & \left( \frac{r}{b} \right)^{2(1-R)} \E \left( w_s^{2(1-R)} \right)
\\
& \leq & \left(\frac{r}{b} \right)^{2(1-R)} w_0^{2(1-R)} \exp(\tilde{b} s)
\eeq
using the bound given by \eqref{3} taking $p=2(1-R)$. Substituting this into \eqref{finite qvariation} gives 
\beq
\E \langle Y \rangle_t  & \leq & K^2 \kappa^2 \left( \frac{r}{b} \right)^{2(1-R)} w_0^{2(1-R)} \int_0^t \exp(\tilde{b} s) ds
\\
& < & \infty. 
\eeq

\item $R > 1$: We have that $\bar{c}$ is an increasing process and $\bar{c}_0 > 0$ by assumption. Thus 
\[
\bar{c}_s^{\; 2(1-R)} \leq \bar{c}_0^{\; 2(1-R)}.
\]
Substituting this into \eqref{finite qvariation} gives
\beq
\E \langle Y \rangle_t & \leq & K^2 \kappa^2 \int_0^t \bar{c}_0^{\; 2(1-R)} ds
\\
& = & K^2 \kappa^2 \bar{c}_0^{\; 2(1-R)} t 
\\
& < & \infty. 
\\
\eeq

\item $R = 1$: In this case, \eqref{finite qvariation} becomes
\beq
\E \langle Y \rangle_t & \leq &  K^2 \kappa^2 \int_0^t 1 ds 
\\
& = & K^2 \kappa^2 t 
\\
& < & \infty. 
\eeq

\end{itemize} 

In all three cases, $\E \langle Y \rangle_t < \infty$ for all $t \geq 0$ which implies that $Y$ is a martingale under the conjectured optimal control. 
\end{proof}

As a final step, we now address the asymptotic behaviour of the residual term $\E[e^{-\rho t} V(w_t, \bar{c}_t)]$. This is essentially the argument given in Lemma 6 in Dybvig \cite{D}.

\begin{lem}\label{residue} \ \\
Fix $\bar{c}_0 > 0$. For all feasible strategies 
\[
\liminf_{t \rightarrow \infty} \E[e^{-\rho t} V(w_t, \bar{c}_t)] \geq 0.
\]
For the optimal control
\[
\lim_{t \rightarrow \infty} \E[e^{-\rho t} V(w_t, \bar{c}_t)] = 0.
\]
\end{lem}
\begin{proof}
Note that the fact that $V_w \geq 0$ (see proof of Lemma \ref{local}) together with the boundary condition \eqref{b/r_boundary} implies that 
\[
V(w_t, \bar{c}_t) \geq V\left(\frac{b\bar{c}_t}{r}, \bar{c}_t \right) = U\left(b\bar{c}_t \right)/\rho \geq U(b \bar{c}_0)/\rho > -\infty.
\]
Consequently, 
\[
\liminf_{t \rightarrow \infty} \E[e^{-\rho t} V(w_t, \bar{c}_t)] \geq \lim_{t \rightarrow \infty} e^{-\rho t} U(b \bar{c}_0)/\rho= 0.
\]
Now, for the conjectured optimal strategy, we will consider the cases $R>1$, $R=1$, and $R^*<R<1$ separately. For $R>1$, we have  $J(0) = 0$ hence from 
\[
v(x) = \inf_{0<z<z_{b/r}} \{ J(z) + xz \}.
\]
we deduce that $v(x) \leq 0$ for all $x \geq b/r$, which implies that $V \leq 0$ by \eqref{value function}. But we just showed that 
\[
\liminf_{t \rightarrow \infty} \E[e^{-\rho t} V(w_t, \bar{c}_t)] \geq 0
\]
which forces 
\[
\lim_{t \rightarrow \infty} \E[e^{-\rho t} V(w_t, \bar{c}_t)] = 0
\]
for the conjectured optimal control. Now for $R=1$, using the boundary condition given in \eqref{b/r_boundary} and the fact that $V_{\bar{c}} \leq 0$, we have 
\beq
V(w,\bar{c})&  \geq & V\left(w, \frac{rw}{b} \right)
\\
& = & U(rw)/\rho
\\
& = & \frac{1}{\rho}(\log w + \log r)
\eeq
which gives a lower bound for $V$. Also, recall that for $x \equiv \frac{w}{\bar{c}} \geq a$ we automatically increase $\bar{c}$ until $x=a$. Thus, for $\bar{c} \leq w/a$, $V(w, \bar{c}) = V(w, w/a)$. This together with the fact that $V_{\bar{c}} \leq 0$ implies that  
\beq
V(w,\bar{c}) & \leq &  V\left(w, \frac{w}{a}\right)
\\
& = & \frac{1}{\rho} \left( \log w + 1 + \log \rho \right) + A
\eeq
where the final equation is by \eqref{v explicit R=1} and \eqref{1 value function}. Hence, to show that 
\[
\E\left[e^{-\rho t} V(w_t, \bar{c}_t) \right] \rightarrow 0 \text{ as } t \rightarrow \infty
\]
it is enough to show that 
\[
\E \left[ e^{-\rho t} \log w_t \right] \rightarrow 0 \text{ as } t \rightarrow \infty.
\]
Taking the logarithm of \eqref{2} for $p=1$ gives 
\beq
\mc{
\E\left[e^{-\rho t} \log w_t \right] 
} \\
& = & e^{-\rho t} \left( \log w_0 \right) + e^{-\rho t} \E \left[ \int_0^t \left(r + \frac{\theta_s}{w_s} (\mu-r) - \frac{c_s}{w_s} - \frac12 \left(\frac{\theta_s}{w_s} \right)^2 \sigma^2 \right) ds \right]
\\
& & \quad +\: e^{-\rho t} \E \left[\int_0^t \frac{\theta_s}{w_s} \sigma dW_s \right] 
\\
& \leq & e^{-\rho t} \left( \log w_0 + \E \left[ \int_0^t \left(r + \frac{(\mu-r)^2}{2 \sigma^2} \right) ds + \int_0^t \frac{\theta_s}{w_s} \sigma dW_s \right] \right)
\eeq
where the quadratic form $(\mu -r)\frac{\theta_s}{w_s} - \frac{\sigma^2}{2} \left(\frac{\theta_s}{w_s}\right)^2$ was replaced by its largest value $(\mu-r)^2/2\sigma^2$ and $\frac{c_s}{w_s}$ was replaced by 0, a lower bound. Thus 
\beq
\E\left[ e^{-\rho t} \log w_t \right] & \leq & e^{-\rho t} \left( \log w_0 + \E \left[ \int_0^t \left(r + \frac{(\mu-r)^2}{2 \sigma^2} \right) ds + \int_0^t \frac{\theta_s}{w_s} \sigma dW_s \right] \right)
\\
& = & e^{-\rho t} \left( \log w_0 + \left(r + \frac{(\mu -r)^2}{2\sigma^2} \right) t \right) 
\\
& \rightarrow & 0 \quad \text{ as } t \rightarrow \infty
\eeq
where $\frac{\theta_s}{w_s}$ bounded (see proof of Lemma \ref{wealth}) implies that $\E\left[ \int_0^t \left(\frac{\theta_s}{w_s} \right)^2 \sigma^2 ds \right] < \infty$ and therefore we have $\E\left[ \int_0^t \left(\frac{\theta_s}{w_s} \right) \sigma dW_s \right] = 0$. 
\\
\\
Finally for $R^*<R<1$, by the same reasoning, 
\beq 
V(w,\bar{c}) & \geq & V\left(w, \frac{rw}{b} \right)
\\
& = & U(rw) /\rho
\\
& = &  \frac{r^{1-R} w^{1-R}}{\rho(1-R)}.
\eeq
We also have
\beq
V(w,\bar{c}) & \leq & V\left(w, \frac{w}{a} \right) 
\\
& = & U(w) \left[ \frac{1}{-A(1-R') }\right]^{-1/R'} 
\\
& = & \frac{w^{1-R}}{1-R} \left[ \frac{1}{-A(1-R') }\right]^{-1/R'}
\eeq
where the first equality is by \eqref{v explicit Rnot1} and \eqref{value function}. Hence, to show that 
\[
\E\left[e^{-\rho t} V(w_t, \bar{c}_t) \right] \rightarrow 0 \text{ as } t \rightarrow \infty
\]
it is enough to show that 
\[
\E \left[ e^{-\rho t} w_t^{1-R} \right] \rightarrow 0 \text{ as } t \rightarrow \infty.
\]
Taking $p=1-R$ in \eqref{2} gives
\beq
\mc{
\E \left[ e^{-\rho t} w_t^{1-R} \right]
} \\
& = & w_0^{1-R} \E \left[ \exp  \left( \int_0^t (1-R) \left( r + \frac{\theta_s}{w_s} (\mu -r) - \frac{c_s}{w_s} - \frac{R}{2} \left( \frac{\theta_s}{w_s} \right)^2 \sigma^2\right) - \rho \:  ds \right) \right.
\\
& & \times \: \left. \exp\left(\int_0^t (1-R) \frac{\theta_s}{w_s} \sigma dW_s - \frac{1}{2} \int_0^t (1-R)^2 \left( \frac{\theta_s}{w_s} \right) ^2 \sigma^2 ds \right) \right]
\\
& \leq & w_0^{1-R} \E \left[ \exp \left( \int_0^t \left( (1-R) \left(r + \frac{\kappa^2}{2R} \right) - \rho \right) ds \right) \right.
\\
& & \times \: \left. \exp \left( \int_0^t (1-R) \frac{\theta_s}{w_s} \sigma dW_s - \frac12 \int_0^t (1-R)^2 \left( \frac{\theta_s}{w_s} \right)^2 \sigma ^2 ds \right) \right]
\eeq
where the quadratic form $(\mu-r) \frac{\theta_s}{w_s} - \frac{R\sigma^2}{2} \left( \frac{\theta_s}{w_s} \right)^2$ in $\frac{\theta_s}{w_s}$ was replaced by its maximum value $(\mu-r)^2/2\sigma^2R$ and $\frac{c_s}{w_s}$ was replaced by 0, a lower bound. Thus 
\beq
\E \left[ e^{-\rho t} w_t^{1-R} \right] & \leq & w_0^{1-R} \exp \left( - \left( \rho - (1-R) \left( r + \frac{\kappa^2}{2R} \right) \right) t \right)
\\
& = & w_0^{1-R} \exp(-R \gamma_M t)
\\
& \rightarrow & 0 \quad \text{ as } t \rightarrow \infty
\eeq
since the stochastic eponential is a supermartingale thus has expectation less or equal to 1 and because $\gamma_M$ (defined in \eqref{gamma_M}) is strictly positive by assumption (see section \ref{marketmodeldraw}). 

\end{proof}

We are now finally ready to provide a proof of the verification theorem, Theorem \ref{verify}. 
\begin{proof}[Proof of Theorem \ref{verify}]
To prove optimality, we need to show that for the optimal control
\[
V(w_0, \bar{c}_0) = \E \left[ \int_{t=0}^\infty e^{-\rho t} U(c_t) dt \right]
\]
and also that for any other feasible strategy, $(\theta, c)$,
\[
V(w_0, \bar{c}_0) \geq  \E \left[ \int_{t=0}^\infty e^{-\rho t} U(c_t) dt \right]. 
\]
From Lemma \ref{nonlocal}, we have that for the optimal control, $Y$ is a martingale which gives
\beq
V(w_0, \bar{c}_0) & = & Y_0
\\
& = & \lim_{t \rightarrow \infty} \E[Y_t]
\\
& = & \lim_{t \rightarrow \infty} \E \left[ \int_{s=0}^t e^{- \rho s} U(c_s) ds + e^{-\rho t} V(w_t, c_t) \right]
\\
& = & \lim_{t \rightarrow \infty} \E \left[ \int_{s=0}^t e^{- \rho s} U(c_s) ds \right] + \lim_{t \rightarrow \infty} \E [e^{-\rho t} V(w_t, c_t)]
\\
& = & \E \left[ \int_{t=0}^\infty e^{-\rho t} U(c_t) dt \right]
\eeq
where exchanging the order of the expectation and the limit is justified by $U(c_s) \geq U(b \bar{c}_0) > -\infty$. We also used the result $\lim_{t \rightarrow \infty} \E [e^{-\rho t} V(w_t, c_t)] = 0$ which was obtained in Lemma \ref{residue}.
\\
\\
To complete the proof observe that by Lemma \ref{nonlocal}, for any feasible strategy, $Y$ is a supermartingale, hence 
\beq
V(w_0, \bar{c}_0) & = & Y_0
\\
& \geq & \lim_{t \rightarrow \infty} \E[Y_t]
\\
& = & \lim_{t \rightarrow \infty} \E\left[ \int_{s=0}^t e^{- \rho s} U(c_s) ds + e^{-\rho t} V(w_t, c_t) \right]
\\
& = & \lim_{t \rightarrow \infty} \E \left[ \int_{s=0}^t e^{- \rho s} U(c_s) ds \right] + \lim_{t \rightarrow \infty} \E [e^{-\rho t} V(w_t, c_t)]
\\
& \geq & \E \left[ \int_{t=0}^\infty e^{-\rho t} U(c_t) dt \right]
\eeq
where exchanging the order of the expectation and the limit is justified by Fatou's lemma (or because $U(c_s) \geq U(b \bar{c}_0) > -\infty$), and we used \mbox{Lemma \ref{residue}} to obtain $\lim_{t \rightarrow \infty} \E [e^{-\rho t} V(w_t, c_t)] \geq 0$.
\\
\\
Hence, we have shown that our conjectured solution is optimal. 
\end{proof} 

\section{The problem is ill-posed for $R \leq R^*$} \label{illposed}

In the standard Merton problem \cite{M}, one observes that for $R \leq R^*$ (for $R^*$ as defined in \eqref{Rstar}), it is possible to find strategies that give the investor infinite expected utility. We observe the same scenario in the case we consider here. The Merton problem with a drawdown constraint on consumption is well-posed if and only if $R > R^*$. In the previous section, we presented and verified the optimal solution for $R > R^*$. Now, for completeness, we will demonstrate a class of strategies that give infinite expected utility if we take $R \leq R^*$. 

\begin{prop}
For $R \leq R^*$, the Merton problem with a drawdown constraint on consumption is ill-posed. That is to say, it is possible to find investment and consumption strategies that give the investor infinite expected utility. 
\end{prop} 
\begin{proof}
We want to show that for $R \leq R^*$, we can choose our investment and consumption strategies to make our investment objective 
\[
\E \left[ \int_0^\infty e^{-\rho t} U(c_t) dt \right]
\]
infinite. We will choose controls such that consumption is non-decreasing. This corresponds to taking $b=1$ in the drawdown constraint, and such a strategy would then clearly work for any $0 < b < 1$ as well. 
\\
\\
Let $\theta_t = \pi_M(w_t - \frac{\lambda\bar{w}_t}{r})$ where $\bar{w}_t = \max_{0\leq s \leq t} w_s$, and $\pi_M = \frac{\mu-r}{\sigma^2 R}$ is the so-called Merton ratio. This is similar to what we see in the standard Merton problem \cite{M}, where the optimal investment strategy is to invest $\pi_M w_t$ in the risky stock, for $\pi_M$ as just defined.
\\
\\
In terms of consumption, let $c_t = \lambda \bar{w}_t$ for $\lambda > 0$ which we will specify later. Substituting this into our wealth equation \eqref{wealth equation} gives 
\[
dw_t = \left( w_t - \frac{\lambda \bar{w}_t}{r} \right) \left[ \left( r + \frac{\kappa^2}{R}  \right)dt + \frac{\kappa}{R} dW_t \right]
\]
where $\kappa = \frac{\mu-r}{\sigma}$ as defined previously. We want to get an explicit solution for $\bar{w}_t$ because this will enable us to calculate our investment objective. To do this, we will use the following argument by Cvitani\'c and Karatzas in \cite{CK}. From the above SDE, we obtain
\[
d\left(w_t - \frac{\lambda \bar{w}_t}{r} \right) = \left(w_t - \frac{\lambda \bar{w}_t}{r} \right) \left[ \left( r + \frac{\kappa^2}{R}  \right)dt + \frac{\kappa}{R} dW_t \right] - \frac{\lambda}{r} d \bar{w}_t.
\]
For convenience, let $\alpha = \lambda/r$ and define  
\[
\hat{w}_t = \left(w_t - \alpha \bar{w}_t \right) \bar{w}_t ^{\frac{\alpha}{1-\alpha}}.
\]
By It\^o's formula
\beq 
d\hat{w}_t & = & (w_t - \alpha \bar{w}_t)d\bar{w}_t^{\frac{\alpha}{1-\alpha}} + \bar{w}_t^{\frac{\alpha}{1-\alpha}} d(w_t - \alpha \bar{w}_t) + d \left\langle w_t - \alpha \bar{w}_t, \bar{w}_t^{\frac{\alpha}{1-\alpha}} \right\rangle 
\eeq
but the last term is zero because $\bar{w}_t^{\frac{\alpha}{1-\alpha}}$ is increasing so has finite variation. Hence we obtain
\beq
d\hat{w}_t & = & (w_t - \alpha \bar{w}_t) \left( \frac{\alpha \bar{w}_t^{\frac{\alpha}{1-\alpha}-1}}{1-\alpha}  d\bar{w}_t \right) 
\\
& & +\: \bar{w}_t^{\frac{\alpha}{1-\alpha}} \left\{ (w_t - \alpha\bar{w}_t) \left[ \left( r + \frac{\kappa^2}{R}  \right)dt + \frac{\kappa}{R} dW_t \right] - \alpha d\bar{w}_t \right\}
\\
& = & (w_t - \alpha \bar{w}_t) \bar{w}_t^{\frac{\alpha}{1-\alpha}}\left[ \left( r + \frac{\kappa^2}{R}  \right)dt + \frac{\kappa}{R} dW_t \right] + \frac{\alpha \bar{w}_t^{\frac{\alpha}{1-\alpha} - 1}}{1- \alpha} [(w_t- \bar{w}_t) d\bar{w}_t]
\eeq
and the last term is zero by the definition of $\bar{w}_t$. Therefore, we get 
\[
d \hat{w}_t = \hat{w}_t\left[ \left( r + \frac{\kappa^2}{R}  \right)dt + \frac{\kappa}{R} dW_t \right]
\]
which does not depend on $\alpha$. We can solve the above SDE explicitly to get 
\[
\hat{w}_t = (1-\alpha)w_0^{\frac{1}{1-\alpha}} \exp \left[ \left(r + \frac{\kappa^2}{R} \right) t  + \frac{\kappa}{R} W_t -  \frac{\kappa^2 t}{2 R^2}  \right]
\]
where we let our initial wealth be $w_0$. From the definition of $\hat{w}_t$, we have that 
\begin{equation}
\max_{0 \leq s \leq t} \hat{w}_s = (\bar{w}_t - \alpha \bar{w}_t) \bar{w}_t^{\frac{\alpha}{1-\alpha}} =
(1-\alpha) \bar{w}_t^{\frac{1}{1-\alpha}}. \label{tilde max 1}
\end{equation}
Define 
\[
Y_t = \exp \left[ \frac{\kappa}{R} W_t -  \frac{\kappa^2 t}{2 R^2}  \right]
\]
and denote 
\[
\bar{Y}_t = \max_{0 \leq s \leq t} Y_s.
\]
Then we can rewrite $\hat{w}_t$ as  
\[
\hat{w}_t = (1-\alpha)w_0^{\frac{1}{1-\alpha}} e^{\left(r + \frac{\kappa^2}{R}\right)t} Y_t
\]
and so
\begin{equation}
\max_{0 \leq s \leq t} \hat{w}_s = (1-\alpha)w_0^{\frac{1}{1-\alpha}} e^{\left(r + \frac{\kappa^2}{R}\right)t} \bar{Y}_t \label{tilde max 2}
\end{equation}
since we will choose $\lambda$ so that $1 - \alpha \geq 0$. Equating \eqref{tilde max 1} and \eqref{tilde max 2} gives 
\[
\bar{w}_t = w_0 e^{(1-\alpha)\left(r + \frac{\kappa^2}{R}\right)t} \bar{Y}_t^{1-\alpha}.
\]
We want to calculate our investment objective which is 
\beq
\mc{
\E\left[ \int_0^\infty e^{-\rho t} U(c_t) dt \right]
} \\
& = & \int_0^\infty e^{-\rho t} \E[U(\lambda \bar{w}_t)] dt 
\\
& = & \int_0^\infty e^{-\rho t} \left(\frac{\lambda^{1-R}}{1-R}\right) \E\left( \bar{w}_t^{1-R} \right) dt 
\\
& = & \int_0^\infty \frac{(\lambda w_0)^{1-R}}{1-R} e^{-\rho t + (1-\alpha)(1-R)\left(r+ \frac{\kappa^2}{R}\right) t} \E \left[ \bar{Y}_t^{(1-\alpha)(1-R)} \right] dt
\\
& \geq & \int_0^\infty \frac{(\lambda w_0)^{1-R}}{1-R} e^{-\rho t + (1-\alpha)(1-R)\left(r+ \frac{\kappa^2}{R}\right) t} dt
\eeq
since $\E \left[ \bar{Y}_t^{(1-\alpha)(1-R)} \right] \geq 1$. This is because $\bar{Y}_t \geq 1$ almost surely and we have  $(1-\alpha)(1-R) \geq 0$ because the feasibility condition 
\[
rw_t \geq 1 \times \bar{c}_t \Rightarrow r w_t \geq \lambda \bar{w}_t \Rightarrow w_t \geq \alpha \bar{w}_t
\]
implies that we must have $0 \leq \alpha \leq 1$, and since $R \leq R^* < 1$ by assumption, we have that $1-R > 0$.  
\\
\\
Now since $R \leq R^*$ or equivalently $\gamma_M \leq 0$ (see \eqref{gamma_M}), as explained in section \ref{marketmodeldraw}, we know that 
\[
\rho + (R-1)\left(r + \frac{\kappa^2}{2R} \right) \leq 0
\]
which implies that 
\[
-\rho + (1-\alpha)(1-R)\left( r+ \frac{\kappa^2}{R} \right) \geq (1-R) \left( \frac{\kappa^2}{2R} - \alpha r - \frac{\alpha \kappa^2}{R} \right)
\]
so we have 
\[
\E\left[ \int_0^\infty e^{-\rho t} U(c_t) dt \right] \geq \int_0^\infty \left(\frac{(\lambda w_0)^{1-R}}{1-R} \right) e^{(1-R)\left( \frac{\kappa^2}{2R} - \alpha r - \frac{\alpha \kappa^2}{R} \right)t}  dt.
\]
Recall that $\alpha = \lambda / r$. The right-hand side of the above inequality is infinite for 
\beq
0 & < &  \frac{\kappa^2}{2R} - \alpha r - \frac{\alpha \kappa^2}{R} \vspace{3mm}
\\
\Leftrightarrow \quad 0 & < & \frac{\kappa^2}{2R} - \lambda\left(1 + \frac{\kappa^2}{rR} \right) \vspace{3mm}
\\
\Leftrightarrow \quad 0 & < & \lambda < \frac{r\kappa^2}{2rR + 2\kappa^2}.
\eeq 
And one can check that for this choice of $\lambda$ we do not violate the condition $0 \leq \alpha \leq 1$ mentioned above. Therefore, taking $\lambda$ in this range allows the investor to obtain infinite expected utility which shows that the Merton problem with a drawdown constraint on consumption is ill-posed for $R \leq R^*$.

\end{proof}

\section*{Acknowledgements}
The author is very grateful to Prof. Chris Rogers for suggesting this project to the author and for carefully reading through an earlier version of this paper. The author would also like to thank Dr Michael Tehranchi for helpful advice and discussions.

\end{document}